\documentclass[11pt,a4paper]{article}
\usepackage{a4wide}
\usepackage{authblk} \usepackage{amsthm} \usepackage[T1]{fontenc}
\usepackage{ascmac}
\usepackage{amsmath,amssymb}
\usepackage{booktabs}
\usepackage{graphicx}
\usepackage{wrapfig}
\usepackage{array}
\usepackage[algoruled,linesnumbered,algo2e,vlined]{algorithm2e}
\usepackage{url}
\usepackage{multirow}
\usepackage{multicol}
\usepackage{cases}
\usepackage[normalem]{ulem}
\usepackage{boites}
\usepackage{here}
\newtheorem{definition}{Definition}
\newtheorem{theorem}{Theorem}
\newtheorem{lemma}{Lemma}

\newtheorem{observation}{Observation}
\title{Nyldon Factorization of \\ Thue-Morse Words and Fibonacci Words}
\date{}
\author[1]{Kaisei~Kishi}
\author[2]{Kazuki~Kai}
\author[2]{Yuto~Nakashima}
\author[2]{Shunsuke~Inenaga}
\author[3]{Hideo~Bannai}
\affil[1]{Department of Information Science and Technology, Kyushu University, Japan}
  \affil[ ]{\texttt{kishi.kaisei.216@s.kyushu-u.ac.jp}}
\affil[2]{Department of Informatics, Kyushu University, Japan}
  \affil[ ]{\texttt{\{nakashima.yuto.003, inenaga.shunsuke.380\}@m.kyushu-u.ac.jp}}
\affil[3]{M\&D Data Science Center, Institute of Integrated Research, Institute of Science Tokyo, Japan}
  \affil[ ]{\texttt{hdbn.dsc@tmd.ac.jp}}
\begin{document}
\maketitle
\begin{abstract}
    The Nyldon factorization is a string factorization that is a non-decreasing product of Nyldon words.
    Nyldon words and Nyldon factorizations are recently defined combinatorial objects inspired by the well-known Lyndon words and Lyndon factorizations.
    In this paper, we investigate the Nyldon factorization of several words.
    First, we fully characterize the Nyldon factorizations of the (finite) Fibonacci and the (finite) Thue-Morse words.
    Moreover, we show that there exists a non-decreasing product of Nyldon words that is a factorization of the infinite Thue-Morse word.
\end{abstract}

\newcommand{\infF}{\mathcal{F}}
\newcommand{\LF}{\mathit{LF}}
\newcommand{\lnps}[1]{\mathsf{lnps}(#1)}
\newcommand{\TM}{\mathit{TM}}
\newcommand{\tm}{\mathit{tm}}
\newcommand{\infTM}{\mathcal{TM}}
\newcommand{\NF}[1]{\mathit{NF}(#1)}
\newcommand{\TMI}[1]{\mathit{t}_{1}(#1)}
\newcommand{\TMII}[1]{\mathit{t}_{2}(#1)}
\newcommand{\TMIII}[1]{\mathit{t}_{3}(#1)}
\newcommand{\TMIV}[1]{\mathit{t}_{4}(#1)}
\newcommand{\TMV}[1]{\mathit{t}_{5}(#1)}
\newcommand{\TMVI}[1]{\mathit{t}_{6}(#1)}
\newcommand{\TMVII}[1]{\mathit{t}_{7}(#1)}
\newcommand{\TMVIII}[1]{\mathit{t}_{8}(#1)}
\newcommand{\TMIX}[1]{\mathit{t}_{9}(#1)}
\newcommand{\TMX}[1]{\mathit{t}_{10}(#1)}
\newcommand{\TMXI}[1]{\mathit{t}_{11}(#1)}
\newcommand{\TMXII}[1]{\mathit{t}_{12}(#1)}
\newcommand{\TMY}[2]{\mathit{t}_{\mathit{#1}}(#2)}
\newcommand{\A}{\alpha}
\newcommand{\B}{\beta}
\newcommand{\tmmorph}{\tau}

\section{Introduction}
The Lyndon factorization~\cite{ChenFL58:_lyndon_factorization_} of a string is the unique factorization 
such that each factor is a Lyndon word and the sequence is a non-increasing product of them.
More precisely, $w_1, \ldots, w_k$ is the Lyndon factorization of a string $w$ if the sequence satisfies the following three conditions:
\begin{enumerate}
    \item $w = w_1 \cdots w_k$,
    \item $w_i \in \Sigma^+$ is a Lyndon word for all $i$ satisfying $1 \leq i \leq k$,
    \item $w_i \succeq w_{i+1}$ for all $i$ satisfying $1 \leq i < k$.
\end{enumerate}
A string $x$ is said to be a Lyndon word~\cite{lyndon54:_burnside} if $x$ is lexicographically smaller than any of its cyclic rotations (or suffixes).
This is a simple and fundamental definition of Lyndon words.
Lyndon words and Lyndon factorizations are well-known combinatorial objects, 
and there are lots of studies on them and its variants  
(e.g.,~\cite{BONIZZONI2018281,DBLP:conf/mfcs/BonizzoniFRZZ24,Lyndon_application_math,DBLP:conf/stringology/HirakawaNIT21,Lyndon_application_algorithm}).
On the other hand, Lyndon words can be defined in a recursive way:
strings of length 1 are Lyndon words, 
strings of length more than 1 are Lyndon words if there is no factorization into a non-increasing product of multiple Lyndon words.

Nyldon words and Nyldon factorizations are symmetric structures which were introduced by Grinberg's natural question~\cite{Nyldon}: 
How do the (Lyndon) structures and properties change when changing from ``non-increase'' to ``non-decrease''?
Charlier et al.~\cite{Nyldon_Charlier} answered several questions on Nyldon words,
for instance, the Nyldon factorization is also unique for any string and there is a unique Nyldon word for each conjugacy class of primitive words.
Later, Garg~\cite{GARG2021_Nyldon-like-set} resolved several questions introduced by Charlier et al.~\cite{Nyldon_Charlier}.
Thanks to their studies, a lot of nice structures and properties of Nyldon words and relations between Lyndon words and Nyldon words were revealed,
but there are still many open questions on Nyldon words and Nyldon factorizations.
One interesting question is given as follows: Can we define Nyldon words by a simpler way such as definitions of Lyndon words by cyclic rotations or suffixes.\footnote{Note that a class of words defined as the set of largest cyclic rotations is the anti-Lyndon words and as the set of largest suffixes is the inverse Lyndon words.}
Such a property may accelerate not only combinatorial studies but also algorithmic studies on strings by using Nyldon words.

In this paper, to further understand Nyldon words and Nyldon factorizations, 
we consider the Nyldon factorizations of Fibonacci and Thue-Morse words.
Our main contributions are three-fold.
\begin{itemize}
    \item We fully characterize the Nyldon factorization of the finite Fibonacci words (Section~\ref{sec:Fib}).
    \item We fully characterize the Nyldon factorization of the finite Thue-Morse words (Section~\ref{sec:finite-TM}).
    \item We show that there exists a non-decreasing sequence of Nyldon words that is a factorization of the infinite Thue-Morse words (Section~\ref{sec:inf-TM}).
\end{itemize}

Siromoney et al.~\cite{SIROMONEY1994101} generalized the notion of Lyndon words to infinite words, 
and they also showed that any infinite word can be factorized into a unique non-increasing sequence of Lyndon words, finite and/or infinite:
either (1)~the infinite sequence of finite Lyndon words or (2)~the concatenation of a finite sequence of finite Lyndon words and a single infinite Lyndon word.
Melançon fully characterized the Lyndon factorization of the infinite Fibonacci words~\cite{MELANCON1996,MELANCON2000137} 
and Ido and Melançon characterized the Lyndon factorization of the infinite Thue-Morse words~\cite{MELANCON1997}.

Our basic idea for the infinite Thue-Morse words is from \cite{MELANCON1997}.
Our result gives a factorization in type-(1).
However, we have not shown its uniqueness yet.
The difficulty of showing its uniqueness may come from the fact that a natural generalization of Nyldon words for infinite words is not known.
We believe that there is a nice generalization of Nyldon words for infinite cases as well as Lyndon words 
and the generalization may implies the uniqueness of Nyldon factorization for infinite words. \section{Preliminaries}

\subsection*{Strings}
Let $\Sigma$ be an {\em alphabet}.
An element of $\Sigma^*$ is called a {\em string}.
The length of a string $w$ is denoted by $|w|$.
The empty string $\varepsilon$ is the string of length 0.
Let $\Sigma^+$ be the set of non-empty strings,
i.e., $\Sigma^+ = \Sigma^* \setminus \{\varepsilon \}$.
For any strings $x$ and $y$,
let $x \cdot y$ (or sometimes $xy$) denote the concatenation of the two strings.
For a string $w = xyz$, $x$, $y$ and $z$ are called
a \emph{prefix}, \emph{substring}, and \emph{suffix} of $w$, respectively.
They are called a \emph{proper prefix}, a \emph{proper substring}, and a \emph{proper suffix} of $w$
if $x \neq w$, $y \neq w$, and $z \neq w$, respectively.
The $i$-th symbol of a string $w$ is denoted by $w[i]$, where $1 \leq i \leq |w|$.
For a string $w$ and two integers $1 \leq i \leq j \leq |w|$,
let $w[i..j]$ denote the substring of $w$ that begins at position $i$ and ends at
position $j$. For convenience, let $w[i..j] = \varepsilon$ when $i > j$.
Also,
let
$w[..i]=w[1..i]$ and
$w[i..]=w[i..|w|]$ and
$w' = w[1..|w|-1]$.
For a string $w$, let $w^1 = w$ and let $w^k = ww^{k-1}$ for any integer $k \ge 2$.
Also, for a string $w$ and a prefix $x$ of $w$, let $x^{-1}w = w[|x|+1 ..]$ 
and for a string $w$ a suffix $z$ of $w$, let $wz^{-1} = w[..|w|-|z|]$.
A sequence of $k$ strings $w_{1}, \ldots, w_k$ is called a \emph{factorization} of a string $w$ if $w = w_{1} \cdots w_k$.
We call $k$ the size of a \emph{factorization} of a string $w$.
For a binary string $w$, $\overline{w}$ denotes the bit-wise flipped string of $w$ (e.g., $\overline{aab} = bba$ over $\{a,b\}$).
Let $\prec$ denote a (strict) total order on an alphabet $\Sigma$.
A total order $\prec$ on the alphabet induces a total order on the set of strings
called the \emph{lexicographic order} w.r.t.~$\prec$, also denoted as $\prec$, i.e.,
for any two strings $x, y \in \Sigma^*$,
We write $x \prec y$ if and only if $x$ is a proper prefix of $y$,
or, there exists $1 \leq i \leq \min\{|x|,|y|\}$
s.t. $x[1..i-1] = y[1..i-1]$ and $x[i] \prec y[i]$.
We also write $x \preceq y$ if and only if $x \prec y$ or $x = y$.
A mapping $\phi: \Sigma^* \rightarrow \Sigma^*$ is called \emph{morphism} on an alphabet $\Sigma$ 
if $\phi(xy) = \phi(x)\phi(y)$ for any $x, y \in \Sigma^*$. 
Let $\phi^{1}(x) = \phi(x)$ and let $\phi^{k}(x) = \phi(\phi^{k-1}(x))$ for any integer $k \geq 2$.    

\subsection*{Lyndon word and Lyndon factorizations}
A string $w$ is a {\em Lyndon word}~\cite{lyndon54:_burnside} w.r.t. a lexicographic order $\prec$,
if and only if $w \prec w[i..]$ for all $1 < i \leq |w|$, i.e.,
$w$ is lexicographically smaller than all its proper suffixes with respect to $\prec$.
The \emph{Lyndon factorization}~\cite{ChenFL58:_lyndon_factorization_} of a string $w$, denoted by $\LF(w)$,
is a unique factorization $\lambda_1, \ldots, \lambda_m$ of $w$,
such that each $\lambda_i \in \Sigma^+$ is a Lyndon word 
and $\lambda_{i} \succeq \lambda_{i+1}$ for $1 \leq i < m$.
Let $w =  abaabbaabbaab$.
The Lyndon factorization of $w$ is $ab$, $aabb$, $aabb$, $aab$.
Since $\LF(w)$ is unique for any $w$, Lyndon word can be defined recursively as follows: every string $w$ whose length is $1$ is a Lyndon word,
and every string $w$ whose length is longer than $1$ is a Lyndon word if and only if $w$ has no factorization that is a lexicographically non-increasing sequence of Lyndon words of length shorter than $|w|$.

\subsection*{Nyldon word and Nyldon factorizations}
\emph{Nyldon words}~\cite{Nyldon} are defined recursively in a similar way to Lyndon words: 
every string $w$ whose length is $1$ is a Nyldon word, and every string $w$ whose length is longer than $1$ is 
a Nyldon word if and only if $w$ cannot be factorized into shorter Nyldon words.
More precisely, there is no factorization $\gamma_1, \ldots, \gamma_m$ of $w$ that satisfies the following three conditions:
\begin{enumerate}
    \item each $\gamma_i \in \Sigma^+$ is a Nyldon word,
    \item $\gamma_{i} \preceq \gamma_{i+1}$ for all $1 \leq i < m$,
    \item $m \geq 2$.
\end{enumerate}
Notice that there is a lexicographical reversal: $\gamma_{i} \preceq \gamma_{i+1}$ while $\lambda_{i} \succeq \lambda_{i+1}$.
It is known that any string $w$ has a unique factorization that satisfies the first two conditions~\cite{Nyldon_Charlier}.
We call such factorization of a string $w$ the \emph{Nyldon factorization} of $w$ and denote it by $\NF{w}$.

We call a string $x$ a Nyldon proper suffix of $w$ if $x$ is a Nyldon word and $x$ is a proper suffix of $w$,
and a string $y$ the longest Nyldon proper suffix of $w$ if $y$ is a Nyldon proper suffix of $w$ and 
$y$ is longer than any other Nyldon proper suffixes of $w$.
$\lnps{w}$ denotes the longest Nyldon proper suffix of $w$.
For $w = a a b ba$,
$\NF{w}= a, a, b, ba$ and the Nyldon proper suffixes of $w$ are $a$ and $ba$.
Hence, $\lnps{w} = ba$.
Notice that $ba$ is a Nyldon word since the possible factorization $b, a$ is not a Nyldon factorization.

\subsection*{Fibonacci words and Thue-Morse words}
The $k$-th (finite) Fibonacci word $F_k$ over a binary alphabet $\{a, b\}$ is defined as follows:
$F_0 = b$, $F_1 = a$, $F_k = F_{k-1} \cdot F_{k-2}$ for any $k \geq 2$
.
Let $f_k$ be the length of the $k$-th Fibonacci word (i.e., $f_k = |F_k|$).
We also define the infinite Fibonacci words over an alphabet $\{a, b\}$ as $\infF = \lim_{k \to \infty} F_{k}$.

The $k$-th (finite) Thue-Morse word $\TM_k$ over a binary alphabet $\{a, b\}$ is defined as follows:
$\TM_0 = a$,  $\TM_k = \TM_{k-1} \cdot \overline{\TM_{k-1}}$ for any $k \geq 1$
.
It is clear from the definition that $|\TM_k| = 2^k$ holds.
We also define the infinite Thue-Morse words over an alphabet $\{a, b\}$ as $\infTM = \lim_{k \to \infty} \TM_{k}$.
The Thue-Morse morphism $\tmmorph$ over a binary alphabet $\{a, b\}$ is defined as follows:
$\tmmorph(a) = ab$, $\tmmorph(b) = ba$.
We call $\tmmorph$ Thue-Morse morphism since $\infTM = \lim_{n \to \infty} \tmmorph^{n}(a)$. \section{Nyldon Factorization of Fibonacci Words} \label{sec:Fib}
In this section, we fully characterize the Nyldon factorization of the $k$-th Fibonacci word $F_k$.
We start with known properties on Nyldon words.

\begin{lemma}[Theorem~{20} of \cite{Nyldon_Charlier}]\label{lNps}
  Let $w \in \Sigma^+ $, $|w| \geq 2$ and $w = ps$ where $s = \lnps{w}$.
  Then $w$ is a Nyldon word iff $p$ is a Nyldon word and $p \succ s$.
\end{lemma}

\begin{lemma}[Theorem~{13} of \cite{Nyldon_Charlier}]\label{Nyldon-suffix}
  Let $w$ be a Nyldon word.
  For each Nyldon proper suffix $s$ of $w$, we have $s \prec w$.
\end{lemma}

Our main result in this section is given by Theorem~\ref{thm:NylF_finite-Fib}.
It explains that the $k$-th Fibonacci word is always factorized into the sequence of two Nyldon words.

\begin{theorem}\label{thm:NylF_finite-Fib}
  For every $k \geq 2$, the Nyldon factorization of $F_k$ is $a, F_k[2..f_k]$.
\end{theorem}

It is easy to see that the product of strings $a \cdot F_k[2..f_k]$ represents $F_k$, and $a \prec F_k[2..f_k]$ also holds since $F_k[2] = b$.
Hence, we only have to show $F_k[2..f_k]$ to be a Nyldon word in order to prove Theorem~\ref{thm:NylF_finite-Fib}.
We prove it in Lemma~\ref{finbonaccisuffix_nyldon}.
For convenience, let $H_k = F_k[2..f_k]$ for every $k \geq 2$, and $h_k$ denote the length of $H_k$ (i.e., $h_k = f_k - 1)$. 

\begin{lemma}\label{finbonaccisuffix_nyldon}
  (i)~For every $k \geq 2$, $H_k$ and $H_k \cdot a$ are Nyldon words.
  (ii)~For every $k \geq 4$, $\lnps{H_k} = H_{k-2}$ and $\lnps{H_k \cdot a} = H_{k-2} \cdot a$.
\end{lemma}

\begin{proof}
  We prove the two statements by induction on $k$.
  
  \vskip.5\baselineskip
  \noindent{\bf Base case.}
  For $k = 2, 3$ of (i), we can simply check that 
  $H_2 = b$, $H_2 \cdot a = ba$, $H_3 = ba$, and $H_3 \cdot a = baa$ are Nyldon word.
  For $k = 4$, we can also check $\lnps{H_4} = \lnps{baab}= b = H_2$. Hence, $H_4$ is a Nyldon word by the fact that $H_4 \cdot (\lnps{H_4})^{-1} = baa = H_3$ is a Nyldon word and Lemma~\ref{lNps}.
  Moreover, we can also check that $\lnps{H_4 \cdot a} = \lnps{baaba}= ba = H_2 \cdot a$ holds and $H_4 \cdot a = baaba$ is a Nyldon word.
  Thus, the two statements also hold for $k=4$.

  \vskip.5\baselineskip
  \noindent {\bf Induction step.}
  Suppose that the statement holds for any $k$ satisfying $k < j$ for some $j \geq 5$.
  We show that the statements hold for $k = j$ (see also Fig.~\ref{fig:sketch}).
Firstly, we show that $\lnps{H_j} = H_{j-2}$.
  By the definitions of $F_k$ and $H_k$, we can write $H_j = H_{j-1} \cdot a \cdot H_{j-2}$
  and $H_{j-2}$ is a Nyldon word by induction hypothesis.
  Assume on the contrary that there exists a longer Nyldon suffix $x \cdot H_{j-2}$ of $H_j$ for some non-empty string $x$.
  Here, $x$ is also assumed to be the shortest such string.
  Namely, $H_{j-2}$ is the longest Nyldon proper suffix of $x \cdot H_{j-2}$.
  By Lemma~\ref{lNps}, $x$ must be a Nyldon word that satisfies $x \succ H_{j-2}$.
  Moreover, $x$ is a suffix of $H_{j-1} \cdot a$.
  Since $\lnps{H_{j-1} \cdot a} = H_{j-3} \cdot a$ by induction hypothesis, $x$ is also a suffix of $H_{j-3} \cdot a$.
  We can also see that $H_{j-3} \cdot a$ is a prefix $H_{j-2}$ by the definitions of $F_k$ and $H_k$.
  These facts and Lemma~\ref{Nyldon-suffix} imply that $x \preceq H_{j-3} \cdot a \prec H_{j-2}$.
  This contradicts the inequality $x \succ H_{j-2}$.
  Hence, $\lnps{H_j} = H_{j-2}$ holds.
  Moreover, we can see that $H_j$ is a Nyldon word by Lemma~\ref{lNps} 
  since $\lnps{H_j} = H_{j-2}$, $H_{j-1} \cdot a$ is a Nyldon (by induction hypothesis), and $H_{j-2}$ is a prefix of $H_{j-1} \cdot a$ (i.e., $H_{j-1} \cdot a \succ H_{j-2}$).
  In a similar way, we can also show that $\lnps{H_j \cdot a} = H_{j-2} \cdot a$ holds and $H_j \cdot a$ is a Nyldon word.
  Therefore, the statement holds for every $k$.
\end{proof}

\begin{figure}[t]
    \centering
    \includegraphics[keepaspectratio,width=\linewidth]{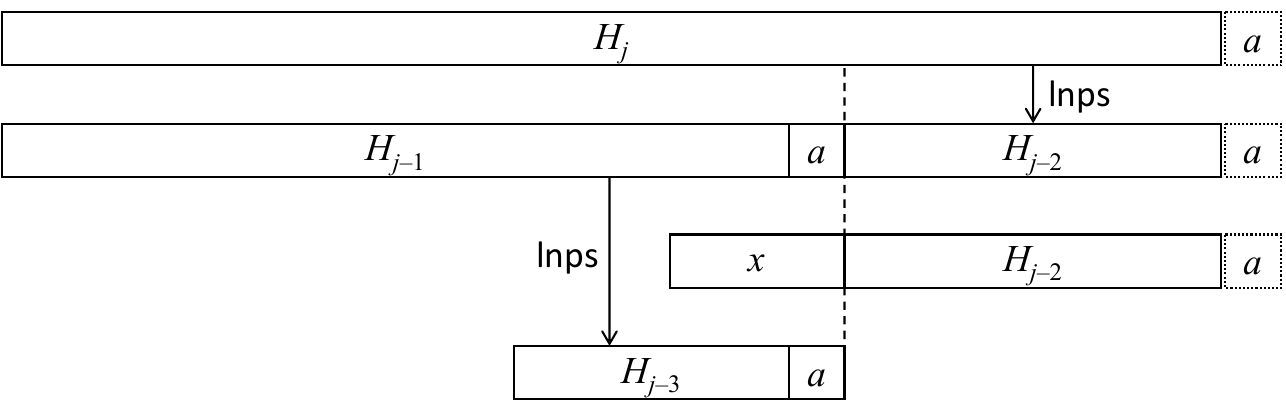}
    \caption{Illustration for proof of Lemma~\ref{finbonaccisuffix_nyldon}.}
    \label{fig:sketch}
\end{figure}

 \section{Nyldon Factorization of Thue-Morse Words}\label{sec:finite-TM}

In this section, we fully characterize the Nyldon factorization of the $k$-th Thue-Morse word $\TM_k$. We show that the factorization can be characterized recursively. Note that $w' = w[1..|w|-1]$ for a string $w$.

\begin{theorem}\label{nyldon_fac_of_tmk}
  Let $\NF{\TM_{2k-1}} = x_1, x_2, \ldots, x_{c_k}$. For every $n \geq 4$,
  \begin{equation*}
    \NF{\TM_n} =
    \begin{cases}
x_1, \ldots, x_{c_{k-1}}, x_{c_k} \cdot {\overline{\TM_{2k-2}'}}, b \cdot \TM_{2k-2} & \text{if $n=2k$,} \\
      x_1, \ldots, x_{c_{k-1}}, x_{c_k} \cdot {\overline{\TM_{2k-2}'}}, b \cdot \TM_{2k-2} \cdot \overline{\TM_{2k}} & \text{if $n=2k+1$.}
    \end{cases}
  \end{equation*}
\end{theorem}

\begin{proof} 
We show this statement by induction on $k$.
First, we can check the statement holds for $k=2$ by the definition of the Nyldon factorization as follows.
\begin{itemize}
  \item $\NF{\TM_3} = a, b, ba, baab$
  \item $\NF{\TM_4} = a, b, ba, baab \cdot baa, b \cdot abba$
  \item $\NF{\TM_5} = a, b, ba, baab \cdot baa, b \cdot abba \cdot baababbaabbabaab$
\end{itemize}
For convenience, we define
\[
  X_k = x_{c_k} \cdot {\overline{\TM_{2k-2}'}}, 
  Y_k = b \cdot \TM_{2k-2}, \text{ and }
  Z_k = b \cdot \TM_{2k-2} \cdot \overline{\TM_{2k}}.
\]
Suppose that the statement holds for $k = j-1$ for some $j \geq 3$.
Now we prove the statement also holds for $k = j$.
To prove the factorization is the Nyldon factorization of the string, we will prove the following three conditions.
\begin{enumerate}
  \item $\TM_{2j} = x_{1} \cdot x_{2} \cdots x_{c_{j-1}} \cdot X_j \cdot Y_j$ 
        and $\TM_{2j+1} = x_{1} \cdot x_{2} \cdots x_{c_{j-1}} \cdot X_j \cdot Z_j$.
  \item $x_{1} \prec x_{2} \prec \ldots \prec x_{c_{j-1}} \prec X_j \prec Y_j$ 
        and $x_{1} \prec x_{2} \prec \ldots \prec x_{c_{j-1}} \prec X_j \prec Z_j$. \item $x_i~(1 \leq i \leq c_{j})$, $X_j$, $Y_j$, $Z_j$ are Nyldon words.
\end{enumerate}

The first property can be shown by induction hypothesis as follows.
\begin{align*}
  x_1 \cdot x_2 \cdots  x_{c_{j-1}} \cdot X_j \cdot Y_j &= x_1 \cdot x_2 \cdots  x_{c_{j-1}} \cdot  x_{c_{j}}  \cdot {\overline{\TM_{2j-2}'}} \cdot b \cdot \TM_{2j-2}\\
                                                                        &= \TM_{2j-1} \cdot {\overline{\TM_{2j-2}'}} \cdot b \cdot \TM_{2j-2}\\
                                                                        &= \TM_{2j-1} \cdot {\overline{\TM_{2j-1}}} = \TM_{2j}.
\end{align*} 
\begin{align*}
  x_1 \cdot x_2 \cdots x_{c_{j-1}} \cdot X_j \cdot Z_j  &= x_1 \cdot x_2 \cdots  x_{c_{j-1}} \cdot  x_{c_{j}}  \cdot {\overline{\TM_{2j-2}'}} \cdot b \cdot \TM_{2j-2} \cdot \overline{\TM_{2j}}\\
                                                                        &= \TM_{2j} \cdot \overline{\TM_{2j}} = \TM_{2j+1}. 
\end{align*} 

Next, we consider the second condition.
By induction hypothesis, $x_1 \prec x_2 \prec \cdots \prec x_{c_j}$ holds.
We can see that $x_{c_j} \prec X_j$ holds since $x_{c_j}$ is a prefix of $X_j$.
In the rest of this part, we show $X_j \prec Y_j \prec Z_j$ instead of proving $X_j \prec Y_j$ and $X_j \prec Z_j$.
Since $Y_j$ is a prefix of $Z_j$, $Y_j \prec Z_j$ holds.
When $j = 2$, we have 
$X_j = b \cdot aab \cdot baa \prec b \cdot abba = Y_j$.
When $j > 2$, we have
$X_j = x_{c_j} \cdot {\overline{\TM_{2j-2}'}}
  = b \cdot \TM_{2j-4} \cdot \overline{\TM_{2j-2}} \cdot {\overline{\TM_{2j-2}'}}$
by induction hypothesis.
To obtain $X_j \prec Y_j$ for every $j > 2$, we can use a part of results in the next section
(e.g., $X_j = \TMXII{j-2}$ and $Y_j = \TMIV{j-1}$ in Lemma~\ref{k_prec_k1}).
Hence $X_j \prec Y_j \prec Z_j$ holds.

Finally, we consider the third condition.
By induction hypothesis, $x_1, x_2, \ldots,$ $x_{c_{j-1}}$ are Nyldon words. In this part, we also use a part of results in the next section.
We can see that $X_j = \TMXII{j-2}$, $Y_j = \TMIV{j-1}$, and $Z_j = \TMX{j-1}$.
These words are Nyldon words by Lemma~\ref{lem:Nyldon-prop} which will be shown in the next section.

Therefore, the given factorizations are Nyldon factorizations of $\TM_{2k}$ and $\TM_{2k+1}$ for every $k \geq 3$.
\end{proof}
 \section{Factorization for infinite Thue-Morse Word}\label{sec:inf-TM}

In this section, we show that there exists an infinite non-decreasing sequence of Nyldon words 
that is a factorization of the infinite Thue-Morse words.
First, we define the sequence in two ways and show that they are equivalent (Lemma~\ref{lem:Nyl-sequences}).

\begin{definition}
  Let $w_1 = a$, $w_2 = b$, $w_3 = ba$, $w_4 = baabbaa$, 
  and $w_n = b \cdot \TM_{2n-8} \cdot \overline{\TM_{2n-6}} \cdot \overline{\TM_{2n-6}'}$ for every $n \geq 5$.
\end{definition}

\begin{definition}\label{def:infTM-morphic-fact}
  Let $\tilde{w}_1 = a$, $\tilde{w}_2 = b$, $\tilde{w}_3 = ba$, $\tilde{w}_4 = baabbaa$, 
$\tilde{w}_5 = b \cdot \TM_{2} \cdot \overline{\TM_{4}} \cdot \overline{\TM_{4}'}$, 
  and $\tilde{w}_{n+1} = (baa)^{-1}\tmmorph^{2}(\tilde{w}_{n}) \cdot baa$ for every $n \geq 5$.
\end{definition}

\begin{lemma} \label{lem:Nyl-sequences}
  For every $n \geq 1$, $w_n = \tilde{w}_n$.
\end{lemma}

\begin{proof}
  We prove it by induction on $n$.
  We can check that the statement holds for every $1 \leq n \leq 5$ by the definitions.
  Suppose that $w_k = \tilde{w}_k$ holds for some $k \geq 5$.
  Then, we have
  \begin{align*}
    \tilde{w}_{k+1} 
    &= (baa)^{-1} \tmmorph^{2}(\tilde{w}_k) \cdot baa = (baa)^{-1} \tmmorph^{2}(w_k) \cdot baa\\
    &= (baa)^{-1} \tmmorph^{2}(b \cdot \TM_{2k-8} \cdot \overline{\TM_{2k-6}} \cdot {\overline{\TM_{2k-6}'}}) \cdot baa\\
&= b \cdot \TM_{2k-6} \cdot \overline{\TM_{2k-4}} \cdot \overline{\TM_{2k-4}'} = w_{k+1}.
  \end{align*}
  Thus, $w_n = \tilde{w}_n$ holds for every $n$.
\end{proof}

In the rest of this section, we use $w_n$ (also for representing $\tilde{w}_n$) to represent the sequence.
The next theorem is a main result of this section that explains the sequence $(w_n)_{n \geq 1}$ is 
an infinite non-decreasing sequence of Nyldon words representing the infinite Thue-Morse words $\infTM$.

\begin{theorem}\label{thm:NylF_TM}
  The sequence of words $(w_n)_{n \geq 1}$ satisfies the following properties;
  \begin{enumerate}
    \item $\infTM = \lim_{n \to \infty} (w_{1} \cdots w_n)$~(Lemma~\ref{lem:factorization-prop}).
    \item $w_n \preceq w_{n+1}$ for every $n \geq 1$~(Lemma~\ref{lem:ordering-prop}).
    \item $w_n$ is a Nyldon word for every $n \geq 1$~(Lemma~\ref{lem:Nyldon-prop}).
  \end{enumerate}
\end{theorem}

\subsection{Factorization property}
We show the first property of Theorem~\ref{thm:NylF_TM}.
The property explains that the sequence $(w_n)_{n \geq 1}$ represents a factorization of $\infTM$.
By the definition of the morphism $\tmmorph$ of the Thue-Morse word,
we prove the following equation.
\begin{equation}\label{eq:infTM-morphism}
  \tmmorph^2 \left( \prod_{n \geq 1} w_n \right) = \prod_{n \geq 1} w_n.
\end{equation}

To prove the equation, we observe relations between some factors and the morphism $\tmmorph^2$.

\begin{observation}\label{obs:property-w}
    The following conditions hold. \begin{enumerate}
      \item $\tmmorph^2(w_1) = \tmmorph^2(a) = abba = w_1w_2w_3$
      \item $\tmmorph^2(w_2) = \tmmorph^2(b) = baab = (baabbaa)(baa)^{-1} = w_4(baa)^{-1}$
      \item $(baa)^{-1}\tmmorph^2(w_3)\tmmorph^2(w_4)(baa) = w_5$
      \item For every $n \geq 5$, $w_n$ has $b$ as a prefix, and $\tmmorph^2(w_n)$ has $baa$ as a prefix.
    \end{enumerate}
\end{observation}

Now we can prove the equation by using the observation as follows.

\begin{lemma} \label{lem:factorization-prop}
  $\infTM = \lim_{n \to \infty} (w_{1} \cdots w_n)$ holds.
\end{lemma}

\begin{proof}
  We show that Equation~\eqref{eq:infTM-morphism} holds.
  \begin{equation*}
    \begin{split}
      \tmmorph^{2} \left( \prod_{n \geq 1} w_n \right)    
        & = \tmmorph^{2}(w_1) \tmmorph^{2}(w_2) \tmmorph^{2}(w_3) \tmmorph^{2}(w_4) (baa) \prod_{n \geq 5} (baa)^{-1}\tmmorph^2(w_n)(baa)\\
        & = w_1 w_2 w_3 \cdot w_4 \cdot w_5 \cdot \prod_{n \geq 5} (baa)^{-1}\tmmorph^2(w_n)(baa)\\
        & = w_1 w_2 w_3 w_4 w_5 \prod_{n \geq 6} w_n = \prod_{n \geq 1} w_n.
\end{split}
  \end{equation*}
  Hence, $\infTM = \lim_{n \to \infty} (w_{1} \cdots w_n)$ holds.
\end{proof}

\subsection{Ordering property}
We show the second property of Theorem~\ref{thm:NylF_TM}.
The property explains that the words in the sequence $(w_n)_{n \geq 1}$ are sorted in lexicographically non-decreasing order.

\begin{lemma} \label{lem:ordering-prop}
  For every $n \geq 1$, $w_n \preceq w_{n+1}$ holds.
\end{lemma}

\begin{proof}
  We show that $w_{n} = x \cdot a \cdot y$ and $w_{n+1} = x \cdot b \cdot z$ for some strings $x$, $y$, and $z$ for every $n \geq 5$.
  Firstly, we can check the statement holds for $n=5$ with $x = babbabaab$.
  We prove this statement by induction on $n$.
  Suppose that $w_{k} = x \cdot a \cdot y$ and $w_{k+1} = x \cdot b \cdot z$ for some strings $x$, $y$, and $z$ for some $k \geq 5$.
  By the definition of $w_k$ and the induction hypothesis,
  $w_{k+1} = (baa)^{-1} \tmmorph^2(w_{k}) \cdot baa = (baa)^{-1} \tmmorph^2(x \cdot a \cdot y) \cdot baa$ and
  $w_{k+2} = (baa)^{-1} \tmmorph^2(w_{k+1}) \cdot baa = (baa)^{-1} \tmmorph^2(x \cdot b \cdot z) \cdot baa$ hold.
  Hence, $w_{k+1} = X a Y$ and $w_{k+2} = X b Z$ hold
  where $X = (baa)^{-1} \tmmorph^2(x)$, $Y = bba \tmmorph^2(y) \cdot baa$, and $Z = aab \tmmorph^2(z) \cdot baa$.
  Thus, the statement holds for every $n \geq 5$.
  This fact also implies that $w_n \preceq w_{n+1}$ for every $n \geq 5$.
  On the other hand, we can also check that $w_1 \preceq w_2 \preceq \ldots \preceq w_5$.
Therefore, $w_n \preceq w_{n+1}$ holds for every $n \geq 1$.
\end{proof}

\subsection{Nyldon property}
We show the third property of Theorem~\ref{thm:NylF_TM}.
Our proof is based on an important combinatorial property of Nyldon words which was explained in Lemma~\ref{lNps}.
We consider the specific Nyldon substrings in $w_n$ that are produced by recursively applying the property (i.e., $\lnps{\cdot}$) to Nyldon substrings.
Let $\Pi_k = \{ t_i(k) \mid i \in [1, 12] \}$
be the set of twelve types of substrings that are Nyldon words for every $k \geq 1$,
where $[l, r]$ represents the set of integers $i$ such that $l \leq i \leq r$.
Every element in $\Pi_k$ is defined as follows.
For convenience, we also represent the substrings by using $A = \TM_{2k}$ and $B = \overline{\TM_{2k}}$.
\begin{alignat*}{2}
  \TMI{k} & = b \cdot {\TM_{2k+1}'} &
    &= b \cdot A \cdot  B'\\
  \TMII{k} &=  b \cdot {\TM_{2k+1}} &
    &= b \cdot A \cdot B\\
  \TMIII{k} &=  b \cdot {\TM_{2k+1}} \cdot {\overline{\TM_{2k}'}} &
    &= b  \cdot A \cdot B \cdot B'\\
  \TMIV{k} &=  b \cdot {\TM_{2k}} &
    &= b \cdot A\\
  \TMV{k} &=  b \cdot \TM_{2k} \cdot \TM_{2k-1} \cdot {\overline{\TM_{2k-2}'}} &
    &= b \cdot A  \cdot {\TM_{2k-1}} \cdot {\overline{\TM_{2k-2}'}}\\
  \TMVI{k} &=  b \cdot {\TM_{2k}} \cdot {\TM_{2k+1}'} &
    &= b \cdot A \cdot  A \cdot  B'\\
  \TMVII{k} &=  b \cdot {\TM_{2k}} \cdot \TM_{2k+1} &
    &= b \cdot A \cdot  A \cdot  B\\
  \TMVIII{k} &=  b \cdot \TM_{2k} \cdot \TM_{2k+1} \cdot {\overline{\TM_{2k}'}} &
    &= b \cdot A \cdot  A \cdot  B \cdot B'\\
  \TMIX{k} &=  b \cdot {\TM_{2k}} \cdot {\overline{\TM_{2k+2}'}} &
    &= b \cdot A \cdot  B \cdot A \cdot A \cdot B'\\
  \TMX{k} &=  b \cdot {\TM_{2k}} \cdot {\overline{\TM_{2k+2}}} &
    &= b \cdot A \cdot  B \cdot A \cdot A \cdot B\\
  \TMXI{k} &=  b \cdot \TM_{2k} \cdot \TM_{2k+1} \cdot {\overline{\TM_{2k+2}'}} &
    &= b \cdot A \cdot  A \cdot  B \cdot B \cdot A \cdot A \cdot B'\\
  \TMXII{k} &=  b \cdot \TM_{2k} \cdot \overline{\TM_{2k+2}} \cdot {\overline{\TM_{2k+2}'}} &
    &= b \cdot A \cdot B \cdot A \cdot A \cdot B \cdot B \cdot A \cdot A \cdot B'
\end{alignat*}
Notice that $\TMXII{k-4} = w_{k}$ for every $k \geq 5$.
Hence, our goal of this part is to show that every string in $\Pi_k$ is a Nyldon word.
This implies the third property of Theorem~\ref{thm:NylF_TM} given by Lemma~\ref{lem:Nyldon-prop}.
Firstly, we introduce properties of strings in $\Pi_k$.

\begin{observation}\label{obs:lexicographical_order_tmk}
 For every $k \geq 1$, $\TMIV{k} \prec \TMV{k} \prec \TMVI{k} \prec \TMVII{k} \prec \TMVIII{k} \prec \TMXI{k} \prec \TMI{k} \prec \TMII{k} \prec \TMIX{k} \prec \TMX{k} \prec \TMXII{k} \prec \TMIII{k}$ holds.
\end{observation}

\begin{lemma} \label{k_prec_k1}
For every $k \geq 1$ and $\ell_1, \ell_2 \in [1, 12]$, $t_{\ell_1}(k) \prec t_{\ell_2}(k+1)$ holds.
\end{lemma}
\begin{proof}
  We first prove $\TMIII{k} \prec \TMIV{k+1}$ for $k \geq 1$.
  It is because 
    \begin{align*}
    \TMIII{k} &= b \cdot {\TM_{2k+1}} \cdot {\overline{\TM_{2k}'}}\\
    &\prec b \cdot {\TM_{2k+1}} \cdot {\overline{\TM_{2k+1}}}\\
    &= b \cdot {\TM_{2k+2}}\\
    &= \TMIV{k+1}.
  \end{align*}
  By this fact and Observation~\ref{obs:lexicographical_order_tmk}, Lemma~\ref{k_prec_k1} holds.
\end{proof}

Now we can show that any string in $\Pi_k$ is a Nyldon word.

\begin{lemma}\label{lem:many-Nyldon}
  For every $k \geq 1$ and $\ell \in [1, 12]$, $\TMY{\ell}{k} \in \Pi_k$ is a Nyldon word.
\end{lemma}

\begin{proof}
  We prove $\TMY{\ell}{k}$ is a Nyldon word for every $k \geq 1$
  and each $\lnps{\TMY{\ell}{k}}$ is a string given in Table~\ref{table_R_k} for every $k \geq 2$.
  We prove these by induction on $k$ and $\ell$.
  
  \vskip.5\baselineskip
  \noindent{\bf Base case.}
  We can check the statements hold for $k = 1, 2$ based on the definitions.
  
  \vskip.5\baselineskip
  \noindent{\bf Induction step.}
  Let $i, j$ be integers satisfying $i>2$ and $j \in [1, 12]$.
  Suppose that the statement holds for every $(k, \ell)$ such that $1 \leq k < i$, or $k = i$ and $\ell<j$.
  We show that the statement also holds for $(k, \ell) = (i, j)$.

  \begin{table}[h]
    \caption{List of $\TMY{\ell}{k} \cdot {\lnps{\TMY{\ell}{k}}}^{-1}$ and $\lnps{\TMY{\ell}{k}}~(k > 1)$}
    \label{table_R_k}
    \centering
    \begin{tabular}{l|l|l}
        \hline
        $\TMY{\ell}{k} \in \Pi_k$ & $\TMY{\ell}{k} \cdot {\lnps{\TMY{\ell}{k}}}^{-1}$ & $\lnps{\TMY{\ell}{k}}$ \\
        \hline
        $\TMI{k}$ & $\TMIII{k-1}$ & $\TMIX{k-1}$   \\
        $\TMII{k}$ & $\TMIII{k-1}$ & $\TMX{k-1}$   \\
        $\TMIII{k}$ & $\TMIII{k-1}$ & $\TMXII{k-1}$   \\
        $\TMIV{k}$ & $\TMIII{k-1}$ & $\TMIV{k-1}$   \\
        $\TMV{k}$ & $\TMIII{k-1}$ & $\TMVIII{k-1}$   \\
        $\TMVI{k}$ & $\TMV{k}$ & $\TMIX{k-1}$   \\
        $\TMVII{k}$ & $\TMV{k}$ & $\TMX{k-1}$   \\
        $\TMVIII{k}$ & $\TMV{k}$ & $\TMXII{k-1}$   \\
        $\TMIX{k}$ & $\TMII{k}$ & $\TMVI{k}$   \\
        $\TMX{k}$ & $\TMII{k}$ & $\TMVII{k}$   \\
        $\TMXI{k}$ & $\TMVIII{k}$ & $\TMVI{k}$   \\
        $\TMXII{k}$ & $\TMI{k}$ & $\TMXI{k}$   \\
        \hline
    \end{tabular}
  \end{table}

Here, we first explain a common overview for all cases.
  After the overview, we give proofs for each case.

  \vskip.5\baselineskip
  \noindent [Sketch of proof for each case]
  Let $t_j(i) = p_0 s_0$ where $p_0 \cdot s_0$ is a factorization given in Table~\ref{table_R_k}.
  Firstly, we check that $t_j(i)$ is the concatenation of $p$ and $s$ in fact by the definition of $t_j(i)$.
  By induction hypothesis, we can see that $s_0$ is a Nyldon suffix of $t_j(i)$. 
  Then, we want to prove $\lnps{t_j(i)} = s_0$.
  Assume on the contrary that there exists a longer Nyldon suffix $x \cdot s$ ($x$ is the shortest such string).
  This assumption implies that $s_0$ is the longest Nyldon proper suffix of $x \cdot s_0$.
  By Lemma~\ref{lNps}, $x$ must be a Nyldon word such that $x \succ s_0$ holds.  
  Moreover, $x$ is a suffix of $s_1 = \lnps{p_0}$.
  Since $s_1 = t_{j'}(i-\ell)$ for some $\ell$ by induction hypothesis,
  $x$ is also a Nyldon suffix of $s_1$.
  By Lemma~\ref{Nyldon-suffix}, $x \preceq s_1$.
  Hence, we have $s_0 \prec x \preceq s_1$.
  However, we can also show that $s_0 \succ s_1$.
  Let $s_0 = t_{i_0}(j_0)$ and $s_1 = t_{i_1}(j_1)$.
  If $j_1 < j_0$ holds, we can use Lemma~\ref{k_prec_k1},
  otherwise, we can show it directly by using Observation~\ref{obs:lexicographical_order_tmk}.
  This contradiction implies that $\lnps{t_j(i)} = s_0$.
  Finally, we want to prove $t_j(i)$ is a Nyldon word.
  We can see that $p_0 \succ s_0$ holds by Lemma~\ref{k_prec_k1} or Observation~\ref{obs:lexicographical_order_tmk}.
  Hence, we can obtain $t_j(i) = p_0 s_0$ by applying Lemma~\ref{lNps}.

  \vskip.5\baselineskip
  \noindent [Proof for $\TMI{i}$]
  We consider the case when $j = 1$.
  Firstly, we show that $\lnps{\TMI{i}} = \TMIX{i-1}$.
  By the definitions of Thue-Morse words, we can write 
  \begin{align*}
    \TMI{i} &=  b \cdot {\TM_{2i+1}'} \\
            &= b \cdot {\TM_{2i-1}} \cdot {\overline{\TM_{2i-2}'}} \cdot b \cdot {\TM_{2i-2}} \cdot {\overline{\TM_{2i}'}} \\
            &= \TMIII{i-1} \cdot \TMIX{i-1}
  \end{align*}
  and the suffix $\TMIX{i-1}$ is a Nyldon word by induction hypothesis.
  Assume on the contrary that there exists a longer Nyldon suffix $x \cdot \TMIX{i-1}$ of $\TMI{i}$ for some non-empty string $x$.
  Here, $x$ is also assumed to be the shortest such string.
  Namely, $\TMIX{i-1}$ is the longest Nyldon proper suffix of $x \cdot \TMIX{i-1}$.
  By Lemma~\ref{lNps}, $x$ must be a Nyldon word such that $x \succ \TMIX{i-1}$ holds.  
  Moreover, $x$ is a suffix of $\TMIII{i-1}$.
  Since $\lnps{\TMIII{i-1}} = \TMXII{i-2}$ by induction hypothesis,
  $x$ is also a Nyldon suffix of $\TMXII{i-2}$.
  By Lemma~\ref{Nyldon-suffix}, $x \prec \TMXII{i-2}$.
  Hence, we have $\TMIX{i-1} \prec x \preceq \TMXII{i-2}$.
  However, this fact contradicts $\TMIX{i-1} \succ \TMXII{i-2}$ by Lemma~\ref{k_prec_k1}.
  This implies that $\lnps{\TMI{i}} = \TMIX{i-1}$.
  Finally, we prove $\TMI{i}$ is a Nyldon word by Lemma~\ref{lNps}.
  We know that $\TMI{i} = \TMIII{i-1} \cdot \TMIX{i-1}$, $\TMIII{i-1} \succ \TMIX{i-1}$ by Observation~\ref{obs:lexicographical_order_tmk}, $\lnps{\TMI{i}} = {\TMIX{i-1}}$,
  and $\TMIII{i-1}$ is a Nyldon word by induction hypothesis.
  Therefore, $\TMI{i}$ is a Nyldon word by these facts and Lemma~\ref{lNps}.

  \vskip.5\baselineskip
  \noindent [Proof for $\TMII{i}$]
  We consider the case when $j = 2$.
  Firstly, we show that $\lnps{\TMII{i}} = \TMX{i-1}$.
  By the definitions of Thue-Morse words, we can write 
  \begin{align*}
    \TMII{i}  &= b \cdot {\TM_{2i+1}}\\
              &= b \cdot {\TM_{2i-1}} \cdot {\overline{\TM_{2i-2}'}} \cdot b \cdot {\TM_{2i-2}} \cdot {\overline{\TM_{2i}}} \\
              &= \TMIII{i-1} \cdot \TMX{i-1}
  \end{align*}
  and the suffix $\TMX{i-1}$ is a Nyldon word by induction hypothesis.
  Assume on the contrary that there exists a longer Nyldon suffix $x \cdot \TMX{i-1}$ of $\TMII{i}$ for some non-empty string $x$.
  Here, $x$ is also assumed to be the shortest such string.
  Namely, $\TMX{i-1}$ is the longest Nyldon proper suffix of $x \cdot \TMX{i-1}$.
  By Lemma~\ref{lNps}, $x$ must be a Nyldon word such that $x \succ \TMX{i-1}$ holds.  
  Moreover, $x$ is a suffix of $\TMIII{i-1}$.
  Since $\lnps{\TMIII{i-1}} = \TMXII{i-2}$ by induction hypothesis,
  $x$ is also a Nyldon suffix of $\TMXII{i-2}$.
  By Lemma~\ref{Nyldon-suffix}, $x \prec \TMXII{i-2}$.
  Hence, we have $\TMX{i-1} \prec x \preceq \TMXII{i-2}$.
However, this fact contradicts $\TMX{i-1} \succ \TMXII{i-2}$ by Lemma~\ref{k_prec_k1}.
This implies that $\lnps{\TMII{i}} = \TMX{i-1}$.
  Finally, we prove $\TMII{i}$ is a Nyldon word by Lemma~\ref{lNps}.
We know that $\TMII{i} = \TMIII{i-1} \cdot \TMX{i-1}$, 
  $\TMIII{i-1} \succ \TMX{i-1}$ by Observation~\ref{obs:lexicographical_order_tmk},
  $\lnps{\TMII{i}} = {\TMX{i-1}}$,
  and $\TMIII{i-1}$ is a Nyldon word by induction hypothesis.
  Therefore, $\TMII{i}$ is a Nyldon word by these facts and Lemma~\ref{lNps}.

  \vskip.5\baselineskip
  \noindent [Proof for $\TMIII{i}$]
  We consider the case when $j = 3$.
  Firstly, we show that $\lnps{\TMIII{i}} = \TMXII{i-1}$.
  By the definitions of Thue-Morse words, we can write 
  \begin{align*}
    \TMIII{i} &= b \cdot {\TM_{2i+1}} \cdot {\overline{\TM_{2i}'}}\\
              &= b \cdot {\TM_{2i-1}} \cdot {\overline{\TM_{2i-2}'}} \cdot b \cdot {\TM_{2i-2}}  \cdot {\overline{\TM_{2i}}} \cdot {\overline{\TM_{2i}'}} \\
              &= \TMIII{i-1} \cdot \TMXII{i-1}
  \end{align*}
  and the suffix $\TMXII{i-1}$ is a Nyldon word by induction hypothesis.
  Assume on the contrary that there exists a longer Nyldon suffix $x \cdot \TMXII{i-1}$ of $\TMIII{i}$ for some non-empty string $x$.
  Here, $x$ is also assumed to be the shortest such string.
  Namely, $\TMXII{i-1}$ is the longest Nyldon proper suffix of $x \cdot \TMXII{i-1}$.
  By Lemma~\ref{lNps}, $x$ must be a Nyldon word such that $x \succ \TMXII{i-1}$ holds.  
  Moreover, $x$ is a suffix of $\TMIII{i-1}$.
  Since $\lnps{\TMIII{i-1}} = \TMXII{i-2}$ by induction hypothesis,
  $x$ is also a Nyldon suffix of $\TMXII{i-2}$.
  By Lemma~\ref{Nyldon-suffix}, $x \prec \TMXII{i-2}$.
  Hence, we have $\TMXII{i-1} \prec x \preceq \TMXII{i-2}$.
However, this fact contradicts $\TMXII{i-1} \succ \TMXII{i-2}$ by Lemma~\ref{k_prec_k1}.
This implies that $\lnps{\TMIII{i}} = \TMXII{i-1}$.
  Finally, we prove $\TMIII{i}$ is a Nyldon word by Lemma~\ref{lNps}.
We know that $\TMIII{i} = \TMIII{i-1} \cdot \TMXII{i-1}$, 
  $\TMIII{i-1} \succ \TMXII{i-1}$ by Observation~\ref{obs:lexicographical_order_tmk},
  $\lnps{\TMIII{i}} = {\TMXII{i-1}}$,
  and $\TMIII{i-1}$ is a Nyldon word by induction hypothesis.
  Therefore, $\TMIII{i}$ is a Nyldon word by these facts and Lemma~\ref{lNps}.

  \vskip.5\baselineskip
  \noindent [Proof for $\TMIV{i}$]
  We consider the case when $j = 4$.
  Firstly, we show that $\lnps{\TMIV{i}} = \TMIV{i-1}$.
  By the definitions of Thue-Morse words, we can write 
  \begin{align*}
    \TMIV{i}  &= b \cdot {\TM_{2i}}\\
              &= b \cdot {\TM_{2i-1}} \cdot {\overline{\TM_{2i-2}'}} \cdot b \cdot {\TM_{2i-2}}  \\
              &= \TMIII{i-1} \cdot \TMIV{i-1}
  \end{align*}
  and the suffix $\TMIV{i-1}$ is a Nyldon word by induction hypothesis.
  Assume on the contrary that there exists a longer Nyldon suffix $x \cdot \TMIV{i-1}$ of $\TMIV{i}$ for some non-empty string $x$.
  Here, $x$ is also assumed to be the shortest such string.
  Namely, $\TMIV{i-1}$ is the longest Nyldon proper suffix of $x \cdot \TMIV{i-1}$.
  By Lemma~\ref{lNps}, $x$ must be a Nyldon word such that $x \succ \TMIV{i-1}$ holds.  
  Moreover, $x$ is a suffix of $\TMIII{i-1}$.
  Since $\lnps{\TMIII{i-1}} = \TMXII{i-2}$ by induction hypothesis,
  $x$ is also a Nyldon suffix of $\TMXII{i-2}$.
  By Lemma~\ref{Nyldon-suffix}, $x \prec \TMXII{i-2}$.
  Hence, we have $\TMIV{i-1} \prec x \preceq \TMXII{i-2}$.
  However, this fact contradicts $\TMIV{i-1} \succ \TMXII{i-2}$ by Lemma~\ref{k_prec_k1}.
  This implies that $\lnps{\TMIV{i}} = \TMIV{i-1}$.
  Finally, we prove $\TMIV{i}$ is a Nyldon word by Lemma~\ref{lNps}.
  We know that $\TMIV{i} = \TMIII{i-1} \cdot \TMIV{i-1}$, 
  $\TMIII{i-1} \succ \TMIV{i-1}$ by Observation~\ref{obs:lexicographical_order_tmk},
  $\lnps{\TMIV{i}} = {\TMIV{i-1}}$,
  and $\TMIII{i-1}$ is a Nyldon word by induction hypothesis.
  Therefore, $\TMIV{i}$ is a Nyldon word by these facts and Lemma~\ref{lNps}.

  \vskip.5\baselineskip
  \noindent [Proof for $\TMV{i}$]
  We consider the case when $j = 5$.
  Firstly, we show that $\lnps{\TMV{i}} = \TMVIII{i-1}$.
  By the definitions of Thue-Morse words, we can write 
  \begin{align*}
    \TMV{i} &= b \cdot \TM_{2i} \cdot \TM_{2i-1} \cdot {\overline{\TM_{2i-2}'}}\\
            &= b \cdot {\TM_{2i-1}} \cdot {\overline{\TM_{2i-2}'}} \cdot b \cdot {\TM_{2i-2}} \cdot {\TM_{2i-1}} \cdot {\overline{\TM_{2i-2}'}} \\
            &= \TMIII{i-1} \cdot \TMVIII{i-1}
  \end{align*}
  and the suffix $\TMVIII{i-1}$ is a Nyldon word by induction hypothesis.
  Assume on the contrary that there exists a longer Nyldon suffix $x \cdot \TMVIII{i-1}$ of $\TMV{i}$ for some non-empty string $x$.
  Here, $x$ is also assumed to be the shortest such string.
  Namely, $\TMVIII{i-1}$ is the longest Nyldon proper suffix of $x \cdot \TMVIII{i-1}$.
  By Lemma~\ref{lNps}, $x$ must be a Nyldon word such that $x \succ \TMVIII{i-1}$ holds.  
  Moreover, $x$ is a suffix of $\TMIII{i-1}$.
  Since $\lnps{\TMIII{i-1}} = \TMXII{i-2}$ by induction hypothesis,
  $x$ is also a Nyldon suffix of $\TMXII{i-2}$.
  By Lemma~\ref{Nyldon-suffix}, $x \prec \TMXII{i-2}$.
  Hence, we have $\TMVIII{i-1} \prec x \preceq \TMXII{i-2}$.
However, this fact contradicts $\TMVIII{i-1} \succ \TMXII{i-2}$ by Lemma~\ref{k_prec_k1}.
This implies that $\lnps{\TMV{i}} = \TMVIII{i-1}$.
  Finally, we prove $\TMV{i}$ is a Nyldon word by Lemma~\ref{lNps}.
We know that $\TMV{i} = \TMIII{i-1} \cdot \TMVIII{i-1}$, 
  $\TMIII{i-1} \succ \TMVIII{i-1}$ by Observation~\ref{obs:lexicographical_order_tmk},
  $\lnps{\TMV{i}} = {\TMVIII{i-1}}$,
  and $\TMIII{i-1}$ is a Nyldon word by induction hypothesis.
  Therefore, $\TMV{i}$ is a Nyldon word by these facts and Lemma~\ref{lNps}.
  
  \vskip.5\baselineskip
  \noindent [Proof for $\TMVI{i}$]
  We consider the case when $j = 6$.
  Firstly, we show that $\lnps{\TMVI{i}} = \TMIX{i-1}$.
  By the definitions of Thue-Morse words, we can write 
  \begin{align*}
    \TMVI{i} &= b \cdot {\TM_{2i}} \cdot {\TM_{2i+1}'}\\
             &= b \cdot {\TM_{2i}} \cdot {\TM_{2i-1}} \cdot {\overline{\TM_{2i-2}'}} \cdot b \cdot {\TM_{2i-2}} \cdot {\overline{\TM_{2i}'}} \\
             &= \TMV{i} \cdot \TMIX{i-1}
  \end{align*}
  and the suffix $\TMIX{i-1}$ is a Nyldon word by induction hypothesis.
  Assume on the contrary that there exists a longer Nyldon suffix $x \cdot \TMIX{i-1}$ of $\TMVI{i}$ for some non-empty string $x$.
  Here, $x$ is also assumed to be the shortest such string.
  Namely, $\TMIX{i-1}$ is the longest Nyldon proper suffix of $x \cdot \TMIX{i-1}$.
  By Lemma~\ref{lNps}, $x$ must be a Nyldon word such that $x \succ \TMIX{i-1}$ holds.  
  Moreover, $x$ is a suffix of $\TMV{i}$.
  Since $\lnps{\TMV{i}} = \TMVIII{i-1}$ by induction hypothesis,
  $x$ is also a Nyldon suffix of $\TMVIII{i-1}$.
  By Lemma~\ref{Nyldon-suffix}, $x \prec \TMVIII{i-1}$.
  Hence, we have $\TMIX{i-1} \prec x \preceq \TMVIII{i-1}$.
  However, this fact contradicts $\TMIX{i-1} \prec x \prec \TMVIII{i-1}$ by Observation~\ref{obs:lexicographical_order_tmk}.
  This implies that $\lnps{\TMVI{i}} = \TMIX{i-1}$.
  Finally, we prove $\TMVI{i}$ is a Nyldon word by Lemma~\ref{lNps}.
  We know that $\TMVI{i} = \TMV{i} \cdot \TMIX{i-1}$, $\TMV{i} \succ \TMIX{i-1}$ by Lemma~\ref{k_prec_k1}, $\lnps{\TMVI{i}} = {\TMIX{i-1}}$,
  and $\TMV{i}$ is a Nyldon word by induction hypothesis.
  Therefore, $\TMVI{i}$ is a Nyldon word by these facts and Lemma~\ref{lNps}.

  \vskip.5\baselineskip
  \noindent [Proof for $\TMVII{i}$]
  We consider the case when $j = 7$.
  Firstly, we show that $\lnps{\TMVII{i}} = \TMX{i-1}$.
  By the definitions of Thue-Morse words, we can write 
  \begin{align*}
    \TMVII{i} &= b \cdot {\TM_{2i}} \cdot \TM_{2i+1} \\
              &= b \cdot {\TM_{2i}} \cdot {\TM_{2i-1}} \cdot {\overline{\TM_{2i-2}'}} \cdot b \cdot {\TM_{2i-2}} \cdot {\overline{\TM_{2i}}} \\
              &= \TMV{i} \cdot \TMX{i-1}
  \end{align*}
  and the suffix $\TMX{i-1}$ is a Nyldon word by induction hypothesis.
  Assume on the contrary that there exists a longer Nyldon suffix $x \cdot \TMX{i-1}$ of $\TMVII{i}$ for some non-empty string $x$.
  Here, $x$ is also assumed to be the shortest such string.
  Namely, $\TMX{i-1}$ is the longest Nyldon proper suffix of $x \cdot \TMX{i-1}$.
  By Lemma~\ref{lNps}, $x$ must be a Nyldon word such that $x \succ \TMX{i-1}$ holds.  
  Moreover, $x$ is a suffix of $\TMV{i}$.
  Since $\lnps{\TMV{i}} = \TMVIII{i-1}$ by induction hypothesis,
  $x$ is also a Nyldon suffix of $\TMVIII{i-1}$.
  By Lemma~\ref{Nyldon-suffix}, $x \prec \TMVIII{i-1}$.
  Hence, we have $\TMX{i-1} \prec x \preceq \TMVIII{i-1}$.
However, this fact contradicts $\TMX{i-1} \prec x \prec \TMVIII{i-1}$ by Observation~\ref{obs:lexicographical_order_tmk}.
This implies that $\lnps{\TMVII{i}} = \TMX{i-1}$.
  Finally, we prove $\TMVII{i}$ is a Nyldon word by Lemma~\ref{lNps}.
  We know that $\TMVII{i} = \TMV{i} \cdot \TMX{i-1}$, 
  $\TMV{i} \succ \TMX{i-1}$ by Lemma~\ref{k_prec_k1},
  $\lnps{\TMVII{i}} = {\TMX{i-1}}$,
  and $\TMV{i}$ is a Nyldon word by induction hypothesis.
  Therefore, $\TMVII{i}$ is a Nyldon word by these facts and Lemma~\ref{lNps}.

  \vskip.5\baselineskip
  \noindent [Proof for $\TMVIII{i}$]
  We consider the case when $j = 8$.
  Firstly, we show that $\lnps{\TMVIII{i}} = \TMXII{i-1}$.
  By the definitions of Thue-Morse words, we can write 
  \begin{align*}
    \TMVIII{i}  &= b \cdot \TM_{2i} \cdot \TM_{2i+1} \cdot {\overline{\TM_{2i}'}}\\
                &= b \cdot {\TM_{2i}} \cdot {\TM_{2i-1}} \cdot {\overline{\TM_{2i-2}'}} \cdot b \cdot {\TM_{2i-2}} \cdot {\overline{\TM_{2i}}} \cdot {\overline{\TM_{2i}'}}\\
                &= \TMV{i} \cdot \TMXII{i-1}
  \end{align*}
  and the suffix $\TMXII{i-1}$ is a Nyldon word by induction hypothesis.
  Assume on the contrary that there exists a longer Nyldon suffix $x \cdot \TMXII{i-1}$ of $\TMVIII{i}$ for some non-empty string $x$.
  Here, $x$ is also assumed to be the shortest such string.
  Namely, $\TMXII{i-1}$ is the longest Nyldon proper suffix of $x \cdot \TMXII{i-1}$.
  By Lemma~\ref{lNps}, $x$ must be a Nyldon word such that $x \succ \TMXII{i-1}$ holds.  
  Moreover, $x$ is a suffix of $\TMV{i}$.
  Since $\lnps{\TMV{i}} = \TMVIII{i-1}$ by induction hypothesis,
  $x$ is also a Nyldon suffix of $\TMVIII{i-1}$.
  By Lemma~\ref{Nyldon-suffix}, $x \prec \TMVIII{i-1}$.
  Hence, we have $\TMXII{i-1} \prec x \preceq \TMVIII{i-1}$.
However, this fact contradicts $\TMXII{i-1} \prec x \prec \TMVIII{i-1}$ by Observation~\ref{obs:lexicographical_order_tmk}.
This implies that $\lnps{\TMVIII{i}} = \TMXII{i-1}$.
  Finally, we prove $\TMVIII{i}$ is a Nyldon word by Lemma~\ref{lNps}.
  We know that $\TMVIII{i} = \TMV{i} \cdot \TMXII{i-1}$, 
  $\TMV{i} \succ \TMXII{i-1}$ by Lemma~\ref{k_prec_k1},
  $\lnps{\TMVIII{i}} = {\TMXII{i-1}}$,
  and $\TMV{i}$ is a Nyldon word by induction hypothesis.
  Therefore, $\TMVIII{i}$ is a Nyldon word by these facts and Lemma~\ref{lNps}.
  
  \vskip.5\baselineskip
  \noindent [Proof for $\TMIX{i}$]
  We consider the case when $j = 9$.
  Firstly, we show that $\lnps{\TMIX{i}} = \TMVI{i}$.
  By the definitions of Thue-Morse words, we can write 
  \begin{align*}
    \TMIX{i}  &= b \cdot {\TM_{2i}} \cdot {\overline{\TM_{2i+2}'}} \\
              &= b \cdot {\TM_{2i+1}} \cdot b \cdot {\TM_{2i}} \cdot {\TM_{2i+1}'} \\
              &= \TMII{i} \cdot \TMVI{i}
  \end{align*}
  and the suffix $\TMVI{i}$ is a Nyldon word by induction hypothesis.
  Assume on the contrary that there exists a longer Nyldon suffix $x \cdot \TMVI{i}$ of $\TMIX{i}$ for some non-empty string $x$.
  Here, $x$ is also assumed to be the shortest such string.
  Namely, $\TMVI{i}$ is the longest Nyldon proper suffix of $x \cdot \TMVI{i}$.
  By Lemma~\ref{lNps}, $x$ must be a Nyldon word such that $x \succ \TMVI{i}$ holds.  
  Moreover, $x$ is a suffix of $\TMII{i}$.
  Since $\lnps{\TMII{i}} = \TMX{i-1}$ by induction hypothesis,
  $x$ is also a Nyldon suffix of $\TMX{i-1}$.
  By Lemma~\ref{Nyldon-suffix}, $x \prec \TMX{i-1}$.
  Hence, we have $\TMVI{i} \prec x \preceq \TMX{i-1}$.
However, this fact contradicts $\TMVI{i} \succ \TMX{i-1}$ by Lemma~\ref{k_prec_k1}.
This implies that $\lnps{\TMIX{i}} = \TMVI{i}$.
  Finally, we prove $\TMIX{i}$ is a Nyldon word by Lemma~\ref{lNps}.
We know that $\TMIX{i} = \TMII{i} \cdot \TMVI{i}$, 
  $\TMII{i} \succ \TMVI{i}$ by Observation~\ref{obs:lexicographical_order_tmk},
  $\lnps{\TMIX{i}} = {\TMVI{i}}$,
  and $\TMII{i}$ is a Nyldon word by induction hypothesis.
  Therefore, $\TMIX{i}$ is a Nyldon word by these facts and Lemma~\ref{lNps}.

  \vskip.5\baselineskip
  \noindent [Proof for $\TMX{i}$]
  We consider the case when $j = 10$.
  Firstly, we show that $\lnps{\TMX{i}} = \TMVII{i}$.
  By the definitions of Thue-Morse words, we can write 
  \begin{align*}
    \TMX{i} &= b \cdot {\TM_{2i}} \cdot {\overline{\TM_{2i+2}}} \\
            &= b \cdot {\TM_{2i+1}} \cdot b \cdot {\TM_{2i}} \cdot {\TM_{2i+1}} \\
            &= \TMII{i} \cdot \TMVII{i}
  \end{align*}
  and the suffix $\TMVII{i}$ is a Nyldon word by induction hypothesis.
  Assume on the contrary that there exists a longer Nyldon suffix $x \cdot \TMVII{i}$ of $\TMX{i}$ for some non-empty string $x$.
  Here, $x$ is also assumed to be the shortest such string.
  Namely, $\TMVII{i}$ is the longest Nyldon proper suffix of $x \cdot \TMVII{i}$.
  By Lemma~\ref{lNps}, $x$ must be a Nyldon word such that $x \succ \TMVII{i}$ holds.  
  Moreover, $x$ is a suffix of $\TMII{i}$.
  Since $\lnps{\TMII{i}} = \TMX{i-1}$ by induction hypothesis,
  $x$ is also a Nyldon suffix of $\TMX{i-1}$.
  By Lemma~\ref{Nyldon-suffix}, $x \prec \TMX{i-1}$.
  Hence, we have $\TMVII{i} \prec x \preceq \TMX{i-1}$.
However, this fact contradicts $\TMVII{i} \succ \TMX{i-1}$ by Lemma~\ref{k_prec_k1}.
This implies that $\lnps{\TMX{i}} = \TMVII{i}$.
  Finally, we prove $\TMX{i}$ is a Nyldon word by Lemma~\ref{lNps}.
We know that $\TMX{i} = \TMII{i} \cdot \TMVII{i}$, 
  $\TMII{i} \succ \TMVII{i}$ by Observation~\ref{obs:lexicographical_order_tmk},
  $\lnps{\TMX{i}} = {\TMVII{i}}$,
  and $\TMII{i}$ is a Nyldon word by induction hypothesis.
  Therefore, $\TMX{i}$ is a Nyldon word by these facts and Lemma~\ref{lNps}.

  \vskip.5\baselineskip
  \noindent [Proof for $\TMXI{i}$]
  We consider the case when $j = 11$.
  Firstly, we show that $\lnps{\TMXI{i}} = \TMVI{i}$.
  By the definitions of Thue-Morse words, we can write 
  \begin{align*}
    \TMXI{i}  &=  b \cdot \TM_{2i} \cdot \TM_{2i+1} \cdot {\overline{\TM_{2i+2}'}}\\
              &= b \cdot \TM_{2i} \cdot \TM_{2i+1}\cdot {\overline{\TM_{2i}'}} \cdot b \cdot {\TM_{2i}} \cdot {\TM_{2i+1}'}\\
              &= \TMVIII{i} \cdot \TMVI{i}
  \end{align*}
  and the suffix $\TMVI{i}$ is a Nyldon word by induction hypothesis.
  Assume on the contrary that there exists a longer Nyldon suffix $x \cdot \TMVI{i}$ of $\TMXI{i}$ for some non-empty string $x$.
  Here, $x$ is also assumed to be the shortest such string.
  Namely, $\TMVI{i}$ is the longest Nyldon proper suffix of $x \cdot \TMVI{i}$.
  By Lemma~\ref{lNps}, $x$ must be a Nyldon word such that $x \succ \TMVI{i}$ holds.  
  Moreover, $x$ is a suffix of $\TMVIII{i}$.
  Since $\lnps{\TMVIII{i}} = \TMVII{i-1}$ by induction hypothesis,
  $x$ is also a Nyldon suffix of $\TMVII{i-1}$.
  By Lemma~\ref{Nyldon-suffix}, $x \prec \TMVII{i-1}$.
  Hence, we have $\TMVI{i} \prec x \preceq \TMVII{i-1}$.
However, this fact contradicts $\TMVI{i} \succ \TMVII{i-1}$ by Lemma~\ref{k_prec_k1}.
This implies that $\lnps{\TMXI{i}} = \TMVI{i}$.
  Finally, we prove $\TMXI{i}$ is a Nyldon word by Lemma~\ref{lNps}.
We know that $\TMXI{i} = \TMVIII{i} \cdot \TMVI{i}$, 
  $\TMVIII{i} \succ \TMVI{i}$ by Observation~\ref{obs:lexicographical_order_tmk},
  $\lnps{\TMXI{i}} = {\TMVI{i}}$,
  and $\TMVIII{i}$ is a Nyldon word by induction hypothesis.
  Therefore, $\TMXI{i}$ is a Nyldon word by these facts and Lemma~\ref{lNps}.

  \vskip.5\baselineskip
  \noindent [Proof for $\TMXII{i}$]
  We consider the case when $j = 12$.
  Firstly, we show that $\lnps{\TMXII{i}} = \TMXI{i}$.
  By the definitions of Thue-Morse words, we can write 
  \begin{align*}
    \TMXII{i} &= b \cdot \TM_{2i} \cdot \overline{\TM_{2i+2}} \cdot {\overline{\TM_{2i+2}'}}\\
              &= b \cdot {\TM_{2i+1}'} \cdot b \cdot {\TM_{2i}} \cdot {\TM_{2i+1}} \cdot {\overline{\TM_{2i+2}'}}\\
              &= \TMI{i} \cdot \TMXI{i}
  \end{align*}
  and the suffix $\TMXI{i}$ is a Nyldon word by induction hypothesis.
  Assume on the contrary that there exists a longer Nyldon suffix $x \cdot \TMXI{i}$ of $\TMXII{i}$ for some non-empty string $x$.
  Here, $x$ is also assumed to be the shortest such string.
  Namely, $\TMXI{i}$ is the longest Nyldon proper suffix of $x \cdot \TMXI{i}$.
  By Lemma~\ref{lNps}, $x$ must be a Nyldon word such that $x \succ \TMXI{i}$ holds.  
  Moreover, $x$ is a suffix of $\TMI{i}$.
  Since $\lnps{\TMI{i}} = \TMIX{i-1}$ by induction hypothesis,
  $x$ is also a Nyldon suffix of $\TMIX{i-1}$.
  By Lemma~\ref{Nyldon-suffix}, $x \prec \TMIX{i-1}$.
  Hence, we have $\TMXI{i} \prec x \preceq \TMIX{i-1}$.
  However, this fact contradicts $\TMXI{i} \succ \TMIX{i-1}$ by Lemma~\ref{k_prec_k1}.
  This implies that $\lnps{\TMXII{i}} = \TMXI{i}$.
  Finally, we prove $\TMXII{i}$ is a Nyldon word by Lemma~\ref{lNps}.
  We know that $\TMXII{i} = \TMI{i} \cdot \TMXI{i}$, 
  $\TMI{i} \succ \TMXI{i}$ by Observation~\ref{obs:lexicographical_order_tmk},
  $\lnps{\TMXII{i}} = {\TMXI{i}}$,
  and $\TMI{i}$ is a Nyldon word by induction hypothesis.
  Therefore, $\TMXII{i}$ is a Nyldon word by these facts and Lemma~\ref{lNps}.

  Finally, we obtain this lemma.
\end{proof}

Now, we are ready to prove the following (third) property.

\begin{lemma} \label{lem:Nyldon-prop}
  For every $n \geq 1$, $w_n$ is a Nyldon word.
\end{lemma}

\begin{proof}
  We can see that $w_1 = a$, $w_2 = b$, $w_3 = ba$, and $w_4 = baabbaa$ are Nyldon words.
  Recall that $\TMXII{k-4} = w_{k}$ for every $k \geq 5$ by the definitions.
  By combining with Lemma~\ref{lem:many-Nyldon}, 
  $w_n$ is a Nyldon word for every $n \geq 1$.
\end{proof} \section{Conclusions}
We considered the Nyldon factorizations of Fibonacci and Thue-Morse words.
We gave the Nyldon factorization of the finite Fibonacci words, 
and fully characterize the Nyldon factorization of the finite Thue-Morse words.
We also showed that there exists a non-decreasing sequence of Nyldon words that is a factorization of the infinite Thue-Morse words.
There are numerous studies on factorizations of certain well-known string families, such as Fibonacci words and Thue-Morse words.
For instance, the Lyndon factorization of the infinite Fibonacci words and sturmian words~\cite{MELANCON1996,MELANCON2000137}, the Lyndon factorization of the infinite Thue-Morse words~\cite{MELANCON1997}, the Crochemore factorization of the sturmian words~\cite{berstel2006crochemore}, repetition factorizations of the finite Fibonacci words~\cite{DBLP:journals/mst/InoueMNIBT22,DBLP:conf/spire/KishiNI23}, palindromic Ziv-Lempel and Crochemore factorizations of $n$-bonacci infinite words~\cite{jahannia2019palindromic} have been considered.
Studies on factorizations for specific strings like these can be useful not only for combinatorics on words, but also for algorithmic studies (e.g., analyzing the complexity of algorithms). We hope our new results also contribute to algorithms that use the Nyldon words and the Nyldon factorizations. 
\section*{Acknowledgments}
This work was supported by JSPS KAKENHI Grant Numbers 
JP25K00136 (YN),
JP23K24808, JP23K18466~(SI),
JP24K02899~(HB), 
and JST BOOST, Japan Grant Number JPMJBS2406~(KK).

\bibliographystyle{abbrv}
\bibliography{ref}

\end{document}